\documentclass[10pt,conference]{IEEEtran}
\usepackage[left=0.625in,right=0.625in,top=0.75in]{geometry}

\def\BibTeX{{\rm B\kern-.05em{\sc i\kern-.025em b}\kern-.08em
    T\kern-.1667em\lower.7ex\hbox{E}\kern-.125emX}}

\usepackage{tikz}
\usetikzlibrary{matrix}
\usepackage{graphicx}
\usepackage{epstopdf}
\usepackage{wrapfig}
\usepackage{amsfonts}
\usepackage{makecell}
\usepackage{pdfpages}
\usepackage[utf8]{inputenc} 
\usepackage{url}
\usepackage{ifthen}
\usepackage{cite}
\usepackage[cmex10]{amsmath}
\usepackage{mathtools}
\usepackage{amsmath,amssymb,amsthm}
\usetikzlibrary{shapes,arrows}
\usepackage{float, flushend}

\usepackage{mathrsfs}  
\usepackage{subcaption}

\theoremstyle{definition}
\newtheorem{theorem}{Theorem}

\theoremstyle{definition}
\newtheorem{definition}{Definition}

\DeclareMathOperator\erf{erf}
\DeclareMathOperator\erfc{erfc}
\DeclareMathOperator\Q{Q}

\begin{document}

\title{On the Rate-Exponent Region of Integrated Sensing and Communications With Variable-Length Coding} 

\author{
   \IEEEauthorblockN{Ioannis Papoutsidakis$^1$ and George C. Alexandropoulos$^{1,2}$}
\IEEEauthorblockA{$^1$Department of Informatics and Telecommunications, National and Kapodistrian University of Athens\\
Panepistimiopolis Ilissia, 15784 Athens, Greece}
\IEEEauthorblockA{$^2$Department of Electrical and Computer Engineering, University of Illinois Chicago, IL 60601, USA 
\\
e-mails: \{ipapout, alexandg\}@di.uoa.gr
}
}

\maketitle
\begin{abstract}
This paper considers the achievable rate-exponent region of integrated sensing and communication systems in the presence of variable-length coding with feedback. This scheme is fundamentally different from earlier studies, as the coding methods that utilize feedback impose different constraints on the codewords. The focus herein is specifically on the Gaussian channel, where three achievable regions are analytically derived and numerically evaluated. In contrast to a setting without feedback, we show that a trade-off exists between the operations of sensing and communications. 
\end{abstract}


\section{Introduction}
The trade-off of sensing and communications in systems that jointly offer these services has been an important field of study in recent years~\cite{6Gdisac}. These so-called Integrated Sensing and Communications (ISAC) systems, in most cases, cannot simultaneously maximize the transmission rate and the sensing capability~\cite{FD_MIMO_ISAC}. This fact is highlighted in one of the initial works that study the capacity-distortion region of ISAC for memoryless channels~\cite{8437621}. The understanding of this trade-off is further expanded by subsequent works that focus on multi-terminal settings~\cite{8849242,9785593}. 

Interestingly, there exist scenarios where an underlying trade-off between the sensing and communication operations does not exist. For instance, the authors of \cite{9731514} demonstrated an ISAC paradigm that avoids compromising  between these two services. It was specifically shown that one can simultaneously maximize the achievable rate and the detection error exponent in the presence of additive Gaussian noise. The assumption that enables this optimal performance is the use of constant-power, fixed-length codewords, which are both capacity-achieving and optimal for sensing.

Nevertheless, there are settings with additive Gaussian noise that naturally do not allow the use of constant-power codewords, such as hybrid automatic repeat request (HARQ) with incremental redundancy \cite{9471818} or variable-length coding with feedback \cite{5961844}. These transmission settings utilize codewords that are virtually infinite in length and the transmission is terminated when a requirement is achieved, such as a threshold on the probability of error. The only way to transmit information with symbols of constant power $\alpha$ at each channel use is with the use of the symbols $\{-\sqrt{\alpha},\sqrt{\alpha}\}$ of a binary Pulse-Amplitude Modulation (PAM) scheme, which is sub-optimal for communication over the Gaussian channel. Thus, there exists an apparent trade-off to be studied in this setting.

In this paper, we introduce three achievable rate-exponent regions for an ISAC setting with variable-length coding with feedback. In contrast to previous works that coarsely approximate their derived regions using optimization algorithms, such as the Blahut-Arimoto algorithm \cite{8437621,9785593,10471902}, we focus on well defined distributions that allow for analytical derivations. The numerical evaluation of our analytic framework showcases for the first time that the advantage offered by the considered variable-length codes comes at the cost of a reduced achievable region when sensing is also considered.

\textit{Notation:} Special functions and distributions are introduced and explained as needed, and all logarithms are base $2$. The Gaussian distribution with mean $\mu$ and variance $\sigma^2$ is denoted as $\mathcal{N}(\mu,\sigma^2)$. $\mathbb{E}[\cdot]$ represents the expectation of a random variable, while ${\rm Pr}[\cdot]$ indicates the probability. Calligraphic letters indicate sets, e.g., $\mathcal{X}$.



\section{System Model and Analysis Objective}
\label{probset}
We consider the setup illustrated in Figure \ref{fig1} comprising an encoder, a decoder, and a detector. The aim of the detector is to estimate the multiplicative state $S$ using side information about the transmitted signal $X^n$ (e.g., a monostatic setting with full duplex~~\cite{FD_MIMO_ISAC}), and the aim of the decoder is to estimate the transmitted information $W$. The state remains unchanged during each codeword transmission and the codeword blocklength is denoted with $n$. For $i=1,\ldots,n$, the decoder and detector observe respectively the following channel outputs:
\begin{align}
\tilde{Y}_i &= X_i+\tilde{Z}_i,\label{eq:Yi}\\
Y_i &= SX_i+Z_i. \label{eq:tilde_Yi}
\end{align}
We assume for our analysis that the additive noise $Z_i$ is independent of $\tilde{Z}_i$, and that $Z_i\sim\mathcal{N}(0,\sigma^2)$ and $\tilde{Z}_i\sim\mathcal{N}(0,\tilde{\sigma}^2)$. In both cases, the noise terms are independent and identically distributed over the channel transmissions.
\begin{figure}
\centering
 \scalebox{0.7}{\begin{tikzpicture}[auto, node distance=2cm,>=latex']

    \node (anchor) {};
    \node [draw, rectangle,below of=anchor, minimum height=2em, minimum width=4em] (encoder) {Encoder};
    \node [right of=encoder] (anchor2){};
    \node [above right of=anchor2,node distance=1cm] (anchor3){};
    \node [draw, circle, below right of=anchor2, node distance=1cm] (mult) {$\cdot$};
    \node [draw, circle, right of=mult, node distance=2cm] (sum2) {+};
    \node [draw, rectangle, minimum height=2em, minimum width=4em, right of=sum2, node distance=2cm] (detector) {Detector};
    \node [draw, circle,right of=anchor3, node distance=2cm] (sum1) {+};
    \node [draw, rectangle, minimum height=2em, minimum width=4em, right of=sum1, node distance=2cm] (decoder) {Decoder};

    \draw [] (encoder) -- node {$X^n$} (1.8,-2);
    \draw [->] (sum1) -- node {} (decoder);
    \draw [->] (mult) -- (sum2);
    \draw [->] (sum2) -- (detector);
    \draw [->] (encoder.east) --++(30pt,0pt) |- (mult);    
    \draw [->] (encoder.east) --++(30pt,0pt) |- (sum1);    
    \draw [->] (decoder.north) --++(0pt,40pt) --++(-191pt,0pt) -- (encoder.north) ;

    \node [left of=encoder, node distance=2cm] (input) {$W$};
    \node [above of=sum1, node distance=1cm] (noise1) {$\tilde{Z}^n$};
    \node [right of=decoder, node distance=2cm] (output1) {$\hat{W}$};
    \node [below of=mult, node distance=1cm] (input2) {$S$};
    \node [right of=detector, node distance=2cm] (output2) {$\hat{S}$};
    \node [below of=sum2, node distance=1cm] (noise2) {$Z^n$};
    \node [below of=detector, node distance=1cm] (output3) {$W$};
    \node [above of=anchor3, node distance=2cm] {\footnotesize\text{\quad\quad\quad\quad   Feedback}};

    \draw [->] (input) -- (encoder);
    \draw [->] (noise1) -- (sum1);
    \draw [->] (input2) -- (mult);
    \draw [->] (noise2) -- (sum2);
    \draw [->] (decoder) -- (output1);
    \draw [->] (detector) -- (output2);
    \draw [->] (output3) -- (detector);

\end{tikzpicture}}
\caption{The considered ISAC setting with feedback.}
\label{fig1}
\end{figure}
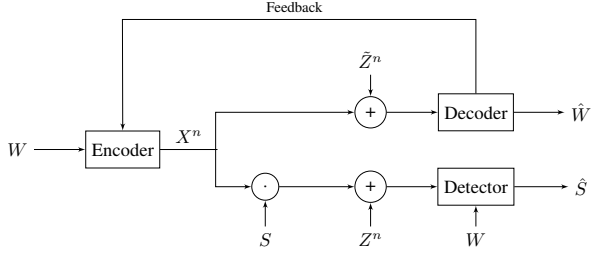
We also assume the utilization of Variable-Length Feedback (VLF) codes, as discussed in \cite{5961844}. It is well known that, even though feedback does not increase capacity, it provides adaptability and significant faster convergence of the error probability to zero compared to fixed blocklength codes. The family of VLF coding schemes is constructed with codewords of very large length, which can be considered infinite for mathematical analysis.

The transmitted variable-length codewords must satisfy the following constraint with $P$ indicating the maximum transmission power (this parameter will be defined in more mathematically convenient forms in the sections that follow):
\begin{align}
\lim_{n \to \infty} \frac{1}{n}\sum_{i=1}^n X_i^2 \leq P.
\end{align}
Hence, the Signal-to-Noise Ratios (SNRs) of the sensing and communication operations are defined as $\text{SNR}_1 \triangleq \frac{P}{\tilde{\sigma}^2}$ and $\text{SNR}_2 \triangleq \frac{P}{\sigma^2}$, respectively. Formally, an $(\ell,M)$ VLF code, where $\ell\in\mathbb{R}^+$ and $M\in\mathbb{Z}^+$, is defined as follows:

\begin{enumerate}
\item A sequence of encoders $f_n:\{1,\ldots,M\} \rightarrow \mathcal{X}$ defining the channel inputs $X_n=f_n(W)$, where $W\in\{1,\ldots,M\}$ is the equiprobable message.
\item A sequence of decoders $g_n:\tilde{\mathcal{Y}}^{n} \rightarrow \{1,\ldots,M\}$ providing the best estimate of $W$ at each channel use $n$.

\item A sequence of detectors $q_n:\{1,\ldots,M\}\times\mathcal{Y}^{n} \rightarrow \mathcal{S}$ providing the best estimate of $S$ at each channel use $n$. 

\item A non-negative integer-valued random variable $\tau$, which represents a stopping time of the filtration $\mathscr{G}_n=\sigma\{Y_1,\ldots,Y_n\}$, for which it holds $\mathbb{E}[\tau]=\ell$.
\end{enumerate}
Essentially, the random variable $\tau$ represents the channel use at which the transmission is completed. Thus, the average decoding error probability is defined as:
\begin{align}
\omega_\ell \triangleq {\rm Pr}[g_\tau(\tilde{Y}_1, \dots, \tilde{Y}_\tau)\neq W],
\end{align}
and the average detection error probability is given by:
\begin{align}
\epsilon_\ell \triangleq {\rm Pr}[q_\tau(W,Y_1, \dots, Y_\tau)\neq S].
\end{align}

We now give the definition of the rate-exponent region which is used to analyze the trade-off between the operations of sensing and communications as in \cite{9731514}. Specifically, this region involves examining the achievable communication rate versus the rate at which the detection error probability converges to zero as the blocklength approaches infinity.
\begin{definition}
A rate-exponent tuple $(R,E)$ is said to be achievable if there exists an $(\ell,M)$ VLF code for which it holds $\lim_{\ell\to +\infty} \omega_\ell=0$ as well as:
\begin{align}
&\lim_{\ell\to +\infty}\frac{1}{\ell}\log M=R,\\
&\lim_{\ell\to +\infty}\frac{1}{\ell}\log\frac{1}{\epsilon_\ell}=E.
\end{align}
The rate-exponent region $\mathcal{R}$ is the closure of the set of all achievable pairs $(R,E)$, i.e.:
\begin{align}
\mathcal{R} = \text{cl} \{(R,E):(R,E)\text{ is achievable}\}.
\end{align}
\end{definition}

\section{Error Detection Probability}
\label{dter}
The detection error exponent was derived analytically in \cite{9731514} by upper bounding the detection error of the maximum-likelihood detector via the Bhattacharyya bound. There exists, however, an efficient way to evaluate the exact probability of erroneous detection, as will be shown in the sequel. For an arbitrary random variable $S$ and for $i=1,\ldots,n$,  \eqref{eq:tilde_Yi} becomes:
\begin{align}
\frac{Y_i}{X_i} = S+\frac{Z_i}{X_i}
\end{align}
with the left-hand-sided random variable distributed as $\frac{Y_i}{X_i} \sim \mathcal{N}\left(S,\sigma^2/X_i^2\right)$. The maximum likelihood estimator of the mean of normal random variables with different variances is their weighted average with weights being the inverse of their variances \cite[(eq. (4.17)]{bevington2003data}. Therefore, the $n$-dimensional estimation problem can be reduced to one dimension, as:
\begin{align}
Y^\prime &= \frac{\sum_{i=1}^n \frac{X_i^2}{\sigma^2}S+\frac{X_i^2}{\sigma^2}\frac{Z_i}{X_i}}{\sum_{i=1}^n \frac{X_i^2}{\sigma^2}}\Rightarrow \nonumber\\
Y^\prime &= \frac{\sum_{i=1}^n X_i^2S+X_iZ_i}{\sum_{i=1}^n X_i^2}\Rightarrow \nonumber\\
Y^\prime &= S + \frac{\sum_{i=1}^n X_iZ_i}{\sum_{i=1}^n X_i^2}\Rightarrow \nonumber\\
Y^\prime &= S + \frac{\sum_{i=1}^n V_i}{\sum_{i=1}^n X_i^2} \text{, where $V_i\sim\mathcal{N}\left(0,X_i^2\sigma^2\right)$}\Rightarrow \nonumber\\
Y^\prime &= S + \frac{B}{\sum_{i=1}^n X_i^2}\text{, where $B\sim\mathcal{N}\left(0,\sum_{i=1}^nX_i^2\sigma^2\right)$}\Rightarrow \nonumber\\
Y^\prime &= S + Z^\prime \text{, where $Z^\prime\sim\mathcal{N}\left(0,\frac{\sigma^2}{\sum_{i=1}^nX_i^2}\right)$.}
\end{align}

For the binary state detection discussed in \cite{9731514}, we have $S\in\{0,1\}$. Since, the Gaussian probability density function is monotone with distance, a minimum distance detection rule is optimal in terms of minimizing the probability of error. Hence, the average detection error probability at each channel use $n$ is given as follows:
\begin{align}
\begin{split}
\epsilon_n &=  {\rm Pr}[S=0]{\rm Pr}[Z^\prime\geq 0.5] + {\rm Pr}[S=1]{\rm Pr}[Z^\prime\leq-0.5]\\
	&=  {\rm Pr}(Z^\prime\leq -0.5)\\
	&= 	\frac{1}{2}\left(1+\erf\left(-\frac{0.5}{\sqrt{\frac{2\sigma^2}{\sum_{i=1}^nX_i^2}}}\right)\right)\\
	&= 	\frac{1}{2}\erfc\left(\frac{0.5}{\sqrt{\frac{2\sigma^2}{\sum_{i=1}^nX_i^2}}}\right) = \Q\left(\frac{\sqrt{\sum_{i=1}^nX_i^2}}{2\sigma}\right),
\end{split}
\end{align}
implying that this probability can be derived in terms of the Gaussian $\Q$-function, depending only on the transmitted power of the channel input signal. A direct application of the Chernoff bound on the $\Q$-function yields the upper bound presented in~\cite{9731514}, since for $x\geq 0$ holds:
\begin{align}
\Q(x) \leq e^{-\frac{x^2}{2}}.
\label{chernoff}
\end{align}

The resulting upper bound on the probability of erroneous detection is shown to provide a derivation of a tight achievable error exponent in \cite{9731514}. Nevertheless, a much shorter proof based on the properties of the $\Q$-function can be obtained as follows. Let $\phi(\cdot)$ be the probability density function of the standard normal distribution. Using \cite{10.1214/aoms/1177731721}, it holds:
\begin{align}
\frac{x}{x^2+1}\phi(x) \leq \Q(x) \leq \frac{1}{x}\phi(x).
\end{align}
By directly evaluating the exponents of the upper and lower bounds, we can deduce the exponent of the $\Q$-function, as both are equal. This indicates that the exponent is given by (the analytical derivation will be included in the extended journal version of this paper):
\begin{align}\label{eq:E}
\begin{split}
E &= -\lim_{n\to+\infty}\frac{1}{n}\log \Q\left(\frac{\sqrt{\sum_{i=1}^nX_i^2}}{2\sigma}\right)\\
  & = \frac{1}{8}\text{SNR}_2 \log e.
\end{split}
\end{align}

As we are particularly interested in the achievable rate-exponent region, we shall utilize the Chernoff bound in \eqref{chernoff} for the remainder of the paper. This bound shares the same exponent as the $\Q$-function, and its simplicity is convenient for the derivations that follow.

\section{The Rate-Exponent Region}
In this section, we present novel analytical expressions for three achievable rate-exponent regions for our ISAC setting. Specifically, we discuss the corner points of the region and two distinct transmit signal distributions that define achievable regions.

\subsection{Corner Points}
\label{cornpoi}
The codebook of a variable-length code consists of infinite-length codewords. Given this setting, two extreme cases can be considered that define two corner points on the achievable rate-exponent region. The sensing-optimal point is achieved with constant-power codewords, as previously discussed. However, the only way to transmit information and ensure that each symbol provides a constant increment to the transmitted power is by using binary PAM. The maximum achievable rate of this binary-input Gaussian channel can be obtained as~\cite{durisi20-11a}:
\begin{align}
\begin{split}
R_{s} &= \frac{1}{\sqrt{2 \pi}}\int_{-\infty}^\infty e^{\frac{-x^2}{2}}\left(1- \log\left(1+e^{-2\text{SNR}_1-2x\sqrt{\text{SNR}_1}}\right)\right)\text{d}x
\end{split}
\label{bigaus}
\end{align}
with the sensing-optimal exponent given via \eqref{eq:E} as $E_{s}= \frac{1}{8}\text{SNR}_2 \log e$. 
In contrast, for the communication-optimal point, the capacity of the Gaussian channel is achieved when the channel inputs follow the Gaussian distribution, i.e.:
\begin{align}
R_{c} =\frac{1}{2}\log(1+\text{SNR}_1).
\end{align}
To derive the exponent of this communication-optimal point, let $P=1$ without loss of generality. Then, the transmitted power follows a $\chi^2$-distribution, and an upper bound on the probability of detection error is given by:
\begin{align}
\epsilon_n &= \int_0^\infty f_{\chi^2}(x;n) \Q\left(\frac{\sqrt{x}}{2\sigma}\right)\text{d}x \nonumber\\
		&\stackrel{(a)}{\leq} \int_0^\infty f_{\chi^2}(x;n) e^{-\frac{x}{8\sigma^2}}\text{d}x \nonumber \\
		&=\int_0^\infty \frac{1}{2^{\frac{n}{2}}\Gamma\left(\frac{n}{2}\right)}x^{\frac{n}{2}-1}e^{-\frac{x}{2}-\frac{x}{8\sigma^2}}\text{d}x \nonumber \\
		&\stackrel{(b)}{=}\int_0^\infty \frac{1}{2^{\frac{n}{2}}\Gamma\left(\frac{n}{2}\right)}\left(\frac{t}{1+\frac{1}{4\sigma^2}}\right)^{\frac{n}{2}-1}e^{-\frac{t}{2}}\frac{\text{d}t}{1+\frac{1}{4\sigma^2}} \nonumber \\
		&=\left(1+\frac{1}{4\sigma^2}\right)^{-\frac{n}{2}} \int_0^\infty \frac{1}{2^{\frac{n}{2}}\Gamma\left(\frac{n}{2}\right)}t^{\frac{n}{2}-1}e^{-\frac{t}{2}}\text{d}t \nonumber \\
		&\stackrel{(c)}{=}\left(1+\frac{1}{4\sigma^2}\right)^{-\frac{n}{2}}=2^{-n\frac{1}{2}\log\left(1+\frac{1}{4}\text{SNR}_2\right)},
\label{comopt}
\end{align}
where $(a)$ is due to the Chernoff bound in (\ref{chernoff}), $(b)$  is obtained by setting $x = t\big/\left(1+\frac{1}{4\sigma^2}\right)$, and $(c)$ follows from noting that the integrand is the probability density function $f_{\chi^2}(t;n)$, integrated over its support.
Hence, the achievable exponent for the communication-optimal point is:
\begin{align}
E_{c} = \frac{1}{2}\log\left(1+\frac{1}{4}\text{SNR}_2\right).
\end{align}

These two corner points define a time-sharing achievable region, which is numerically evaluated in a later section.
%
%

\subsection{Connecting the Corner Points: Gaussian Mixture Signals}
\label{con1}
Here, we discuss a natural way to define an achievable region by connecting the corner points without time-sharing. Specifically, we study the case where the channel input is the summation of a binary PAM signal and a Gaussian signal, resulting in a Gaussian mixture. 
\begin{theorem}
For any $a\in[0,\infty)$, there exists an achievable rate-exponent tuple $(R^*,E^*)$ such that:
\begin{align}
\begin{split}
R^*&=\frac{1}{2}\log\left(1+\frac{\text{SNR}_1}{1+a}\right)+\frac{1}{\sqrt{2 \pi}}\int_{-\infty}^\infty e^{\frac{-x^2}{2}}\\
  &\quad\quad\times\left(1 - \log\left(1+e^{-2\frac{a\text{SNR}_1}{1+a+\text{SNR}_1}-2x\sqrt{\frac{\text{SNR}_1}{1+a+\text{SNR}_1}}}\right)\right)\text{d}x,\\
E^*&=\frac{1}{2}\log\left(1+\frac{\text{SNR}_2}{4+4a}\right)+\frac{a\text{SNR}_2}{8+8a+2\text{SNR}_2}\log e.
\end{split}
\end{align}
\end{theorem}
\begin{proof}
Without loss of generality, assume that each channel input is the summation of a Gaussian signal that follows $\mathcal{N}(0,1)$ and a binary PAM symbol from $\{-\sqrt{a},\sqrt{a}\}$, for any $a\in[0,\infty)$. Then, the SNRs of the operations of communication and sensing are given respectively by:
\begin{align}
\text{SNR}_1 = \frac{1+a}{\tilde{\sigma}^2},\,\,
\text{SNR}_2 = \frac{1+a}{\sigma^2}.
\end{align}

To evaluate the achievable rate, we utilize expression (\ref{bigaus}) in a setting of successive cancellation decoding since our approach is essentially a superposition coding technique \cite{Gamal_Kim_2011}. Specifically, the PAM symbols are decoded assuming the Gaussian component of the channel input is noise. Then, the Gaussian signal is decoded by subtracting the already decoded PAM symbols. The aggregate rate of this decoding scheme is known to be equivalent to the direct evaluation of the mutual information, hence:
\begin{align}
\begin{split}
R^* &= \frac{1}{2}\log\left(1+\frac{1}{\tilde{\sigma}^2}\right)+\frac{1}{\sqrt{2 \pi}}\int_{-\infty}^\infty e^{\frac{-x^2}{2}}\\
&\quad\quad\quad\quad\quad\quad\times\left[1 - \log\left(1+e^{-2\frac{a}{\tilde{\sigma}^2+1}-2x\sqrt{\frac{a}{\tilde{\sigma}^2+1}}}\right)\right]\text{d}x\\
  &= \frac{1}{2}\log\left(1+\frac{\text{SNR}_1}{1+a}\right)+\frac{1}{\sqrt{2 \pi}}\int_{-\infty}^\infty e^{\frac{-x^2}{2}}\\
  &\quad\quad\times\left[1 - \log\left(1+e^{-2\frac{a\text{SNR}_1}{1+a+\text{SNR}_1}-2x\sqrt{\frac{\text{SNR}_1}{1+a+\text{SNR}_1}}}\right)\right]\text{d}x.
\end{split}
\end{align}

The distribution of the sum of squares of this channel input follows a non-central $\chi^2$-distribution with the probability density function $f_{\chi^2_{NC}}(x; n, \lambda)$, where $n$ denotes the degrees of freedom and $\lambda$ is the non-centrality parameter. Then, by setting $\lambda = an$, we derive (\ref{makr}) (top of next page). 
\begin{figure*}
\rule[1ex]{\textwidth}{0.1pt}
\resizebox{0.9\linewidth}{!}{
  \begin{minipage}{\linewidth}
\begin{align}
\begin{split}
\epsilon_n &= \int_0^\infty f_{\chi^2_{NC}}(x;n,\lambda) \Q\left(\frac{\sqrt{x}}{2\sigma}\right)\text{d}x\\
		&\stackrel{(a)}{\leq} \int_0^\infty f_{\chi^2_{NC}}(x;n,\lambda) e^{-\frac{x}{8\sigma^2}}\text{d}x 
  = \int_0^\infty \frac{1}{2} e^{-\frac{x+\lambda}{2}-\frac{x}{8\sigma^2}}\left(\frac{x}{\lambda}\right)^{\frac{n}{4}-\frac{1}{2}} I_{\frac{n}{2}-1}\left(\sqrt{\lambda x}\right)\text{d}x\\
		&\stackrel{(b)}{=} 
  \left(\frac{1}{1+\frac{1}{4\sigma^2}}\right)^{\frac{n}{4}+\frac{1}{2}} \int_0^\infty \frac{1}{2} e^{-\frac{t+\lambda}{2}}\left(\frac{t}{\lambda}\right)^{\frac{n}{4}-\frac{1}{2}} I_{\frac{n}{2}-1}\left(\sqrt{\frac{\lambda}{1+\frac{1}{4\sigma^2}}t}\right)\text{d}t\\
		&\stackrel{(c)}{=} 
  \left(\frac{1}{1+\frac{1}{4\sigma^2}}\right)^{\frac{n}{2}}e^{\frac{\lambda^\prime}{2}-\frac{\lambda}{2}} \int_0^\infty \frac{1}{2} e^{-\frac{t+\lambda^\prime}{2}}\left(\frac{t}{\lambda^\prime}\right)^{\frac{n}{4}-\frac{1}{2}} I_{\frac{n}{2}-1}\left(\sqrt{\lambda^\prime t}\right)\text{d}t\\
		&\stackrel{(d)}{=} \left(1+\frac{1}{4\sigma^2}\right)^{-\frac{n}{2}}e^{-\lambda\left(\frac{1}{2}-\frac{1}{2+\frac{1}{2\sigma^2}}\right)} = \left(1+\frac{1}{4\sigma^2}\right)^{-\frac{n}{2}}e^{-\frac{\lambda}{8\sigma^2+2}} = \left(1+\frac{\text{SNR}_2}{4(1+a)}\right)^{-\frac{n}{2}}e^{-n\frac{a\text{SNR}_2}{8+8a+2\text{SNR}_2}}\\
		&= 2^{-n\left(\frac{1}{2}\log\left(1+\frac{\text{SNR}_2}{4+4a}\right)+\frac{a\text{SNR}_2}{8+8a+2\text{SNR}_2}\log e\right)}
\end{split}
\label{makr}
\end{align}
 \end{minipage}
}
\rule[1ex]{\textwidth}{0.1pt}
\end{figure*}
In this derivation, $(a)$ is due to the Chernoff bound in (\ref{chernoff}), $(b)$  is obtained by setting $x = t\big/\left(1+\frac{1}{4\sigma^2}\right)$, $(c)$ is due to the assignment $\lambda = \lambda^\prime\left(1+\frac{1}{4\sigma^2}\right)$, and $(d)$ follows from noting that the integrand is the probability density function $f_{\chi^2_{NC}}(t; n, \lambda^\prime)$, integrated over its support. Hence, the achievable exponent is: 
\begin{align}
E^* = \frac{1}{2}\log\left(1+\frac{\text{SNR}_2}{4+4a}\right)+\frac{a\text{SNR}_2}{8+8a+2\text{SNR}_2}\log e.
\end{align}
\end{proof}

\subsection{Connecting the Corner Points: A Novel Signal Distribution}
\label{con2}
A major challenge in analytically evaluating an achievable region is to derive the detection error exponent in closed form, as numerical approximations can be inaccurate or converge slowly. In the previous section, we addressed this challenge by utilizing well-known results and distributions. In the same spirit, this section presents an achievable rate-exponent region based on a channel input distribution that leads to an efficient evaluation of the detection error exponent. 

We begin by defining the channel input distribution. Let $U$, $G$, and $X$ be random variables, where $U$ is a binary random variable with ${\rm Pr}(U = -1) = {\rm Pr}(U = 1) = 0.5$, $G$ follows the $\chi$-distribution with $k$ degrees of freedom, and $X \triangleq UG$.
Hence, the probability density function of $X$ is given by: 
\begin{align}
\begin{split}
f_X(x;k) &= \begin{cases} 
      \frac{1}{2}\frac{1}{2^{\frac{k}{2}-1}\Gamma\left(\frac{k}{2}\right)}\lvert x\rvert^{k-1}e^{-\frac{x^2}{2}}, & x < 0 \\
      \frac{1}{2}\frac{1}{2^{\frac{k}{2}-1}\Gamma\left(\frac{k}{2}\right)}x^{k-1}e^{-\frac{x^2}{2}}, & x \geq 0
   \end{cases}\\
    &=\frac{1}{2^{\frac{k}{2}}\Gamma\left(\frac{k}{2}\right)}\lvert x\rvert^{k-1}e^{-\frac{x^2}{2}},
\end{split}
\label{pdfdchi}
\end{align}
where $\mathbb{E}[X]=0$ due to the independence of the variables $U$ and $G$, and the variance of $X$ is obtained as:
\begin{align}
\begin{split}
\sigma_X^2 &= \mathbb{E}[X^2]-\mathbb{E}[X]^2 = \mathbb{E}[U^2G^2] = \mathbb{E}[G^2]\stackrel{(a)}{=} k
\end{split}
\end{align}
with $(a)$ resulting from the fact that $G^2$ follows the $\chi^2$-distribution. The relationship $X^2=G^2$ reveals the motivation behind the definition of $X$. We have defined a zero-mean distribution that is suitable for communications, with a well-defined distribution for the transmitted power $\sum_{i=1}^nX_i^2$, as the sum of $\chi^2$ random variables follows a $\chi^2$-distribution; we next denote this distribution as $\mathcal{J}(k)$ with parameter $k$ indicating the degrees of freedom.

With the channel input distribution defined, we now present the main result of this part in the following theorem.\begin{theorem}
For any $k\in\mathbb{Z}^+$, there exists an achievable rate-exponent tuple $(R^\dag,E^\dag)$ such that:
\begin{align}
\begin{split}
R^\dag&=h(X+\tilde{Z})-\frac{1}{2}\log\left(\frac{2\pi e k}{\text{SNR}_1}\right),\\
E^\dag&=\frac{k}{2}\log\left(1+\frac{\text{SNR}_2}{4k}\right),
\end{split}
\end{align}
where $X\sim\mathcal{J}(k)$ and $\tilde{Z}\sim\mathcal{N}\left(0,k/\text{SNR}_1\right)$.
\end{theorem}

\begin{proof}
Without loss of generality, we define the SNRs of the communication and sensing operations for any $k \in \mathbb{Z}^+$ as:
\begin{align}
\text{SNR}_1 = \frac{k}{\tilde{\sigma}^2},\,\,\text{SNR}_2 = \frac{k}{\sigma^2}.
\end{align}
Notice that the input of the decoder is of the form $Y = X+\tilde{Z}$ with $\tilde{Z}\sim\mathcal{N}(0,\tilde{\sigma}^2)$, hence, the achievable rate is defined as:
\begin{align}
R^\dag &= I(Y;X) = h(Y)-h(Y|X) = h(X+\tilde{Z})-h(\tilde{Z})\nonumber\\
  &= h(X+\tilde{Z})-\frac{1}{2}\log(2\pi e\tilde{\sigma}^2)\nonumber\\
  &=h(X+\tilde{Z})-\frac{1}{2}\log\left(\frac{2\pi e k}{\text{SNR}_1}\right)
\label{entr}
\end{align}
with $I(\cdot;\cdot)$ being the mutual information, and $h(\cdot)$ and $h(\cdot|\cdot)$ is the entropy and the conditional entropy, respectively.
\begin{figure*}
\centering
    \subfloat[$\text{SNR}_1 =\text{SNR}_2 = 10$ dB.]{\includegraphics[width=0.34\linewidth]{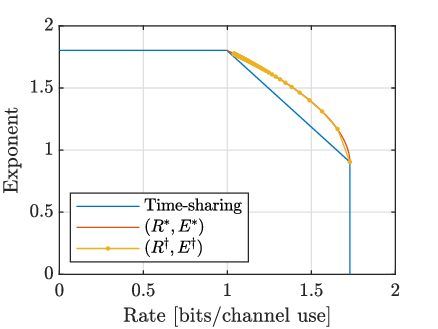} \label{numan2}}\,
    \subfloat[$\text{SNR}_1 =\text{SNR}_2 = 15$ dB.]{\includegraphics[width=0.34\linewidth]{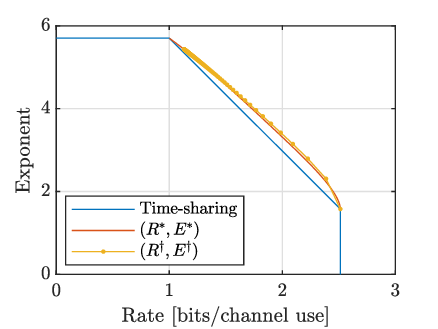} \label{numan3}}\,
    \subfloat[$\text{SNR}_1 =\text{SNR}_2 = 20$ dB.]{\includegraphics[width=0.34\linewidth]{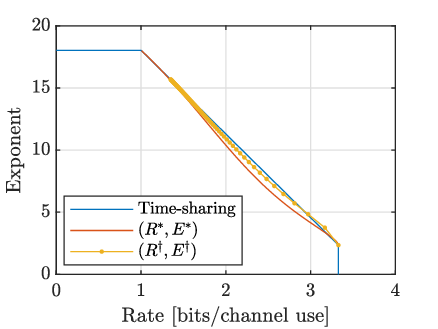} \label{numan4}}\,
    \subfloat[$\text{SNR}_1 =\text{SNR}_2 = 25$ dB.]{\includegraphics[width=0.34\linewidth]{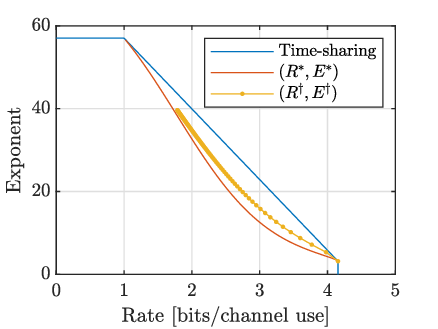} \label{numan5}}
\caption{Achievable rate-exponent regions for various SNR values for the communication ($\text{SNR}_1$) and sensing ($\text{SNR}_2$) operations.}
\label{numan}
\end{figure*}

For the detection error exponent, we provide a derivation similar to that in (\ref{comopt}), based on the fact that the sum of $n$ squared random variables with probability density function (\ref{pdfdchi}) follows the $\chi^2$-distribution with $nk$ degrees of freedom.
\begin{align}
\begin{split}
\epsilon_n &= \int_0^\infty f_{\chi^2}(x;nk) \Q\left(\frac{\sqrt{x}}{2\sigma}\right)\text{d}x\\
		&\stackrel{(a)}{\leq} \int_0^\infty f_{\chi^2}(x;nk) e^{-\frac{x}{8\sigma^2}}\text{d}x \\
		&=\int_0^\infty \frac{1}{2^{\frac{nk}{2}}\Gamma\left(\frac{nk}{2}\right)}x^{\frac{nk}{2}-1}e^{-\frac{x}{2}-\frac{x}{8\sigma^2}}\text{d}x \\
		&\stackrel{(b)}{=}\int_0^\infty \frac{1}{2^{\frac{nk}{2}}\Gamma\left(\frac{nk}{2}\right)}\left(\frac{t}{1+\frac{1}{4\sigma^2}}\right)^{\frac{nk}{2}-1}e^{-\frac{t}{2}}\frac{\text{d}t}{1+\frac{1}{4\sigma^2}}\\
		&\stackrel{(c)}{=}\left(1+\frac{1}{4\sigma^2}\right)^{-\frac{nk}{2}}\\
		&=2^{-n\frac{k}{2}\log\left(1+\frac{\text{SNR}_2}{4k}\right)},
\end{split}
\end{align}
where $(a)$ is due to the Chernoff bound in (\ref{chernoff}), $(b)$  is obtained by setting $x = t\big/\left(1+\frac{1}{4\sigma^2}\right)$, and $(c)$ follows from noting that the integrand is the probability density function $f_{\chi^2}(t;nk)$, integrated over its support. Hence, the achievable exponent is:
\begin{align}
E^\dag = \frac{k}{2}\log\left(1+\frac{\text{SNR}_2}{4k}\right).
\end{align}
\end{proof}


\section{Numerical Results}
\label{neval}
In this section, we numerically evaluate the previously derived analytical expressions  for the achievable rate-exponent region for ISAC with VLF coding for different SNR values. We have particularly considered that both $\text{SNR}_1$ (communication) and $\text{SNR}_2$ (sensing) are set equal to $\{10,15,20,25\}$ dB. For the region $(R^\dag,E^\dag)$ derived for the presented input distribution in Section~\ref{con2}, parameter $k$ for the degrees of freedom of the $\chi^2$-distribution was assigned the values $\{1,2,\ldots,80\}$, since it was found to converge slowly, especially at higher SNRs, making numerical integration challenging.

As illustrated in Fig.~\ref{numan}, at SNR values below $10$ dB, the regions $(R^*,E^*)$ and $(R^\dag,E^\dag)$ are virtually equivalent, wherever the latter is defined. This is not the case at higher SNRs as $(R^\dag,E^\dag)$ becomes larger. However, as SNR increases, both regions become smaller than the time-sharing region defined by the corner points. Notice that a time-sharing region can be defined between any achievable point. For example, in Fig.~\ref{numan4}, a time-sharing region can be defined between $(R^\dag,E^\dag)$ and the sensing-optimal corner point that outperforms the original time-sharing region. Finally, as depicted in Fig.~\ref{numan5} at SNRs greater than $20$ dB, time-sharing between the extreme corner points is the best option. A potential explanation of this phenomenon is that, when the SNR is very large, the noise does not influence the channel output enough so as to approximately follow a Gaussian distribution. Hence, the loss in the achievable rate performance due to sub-optimal signaling becomes greater.

\section{Conclusion and Future Work}
\label{concl}
This paper presented novel analytic results for the achievable rate-exponent region of an ISAC system that utilizes variable-length codes with feedback. The nature of such coding schemes results in a trade-off that does not exist in the fixed blocklength case. The main motivation for considering variable-length codes was their improved non-asymptotic performance compared to fixed-length codes, since the communication error probability exhibits faster convergence to zero. We presented novel expressions for the rate-exponent tuple for both Gaussian mixture input signals and a channel input distribution that leads to an efficient evaluation of the detection error exponent. The numerical evaluation of the derived analytical formulas unveiled that the advantage offered by variable-length codes comes at the cost of a reduced achievable rate-exponent region when sensing is also considered.

Extensions of the present work could examine settings with rate-distortion regions, as in \cite{8437621} and \cite{9785593}. Interestingly, the Blahut-Arimoto algorithm used in these works produces channel input distributions very similar to those in Theorems 1 and 2 of this paper \cite{10471902}. Additionally, we intend to derive a generalized version of Theorem 2 for any $k\in[0,+\infty)$.

\bibliography{ISAC_exponent}
\bibliographystyle{ieeetr}

\end{document}